\documentclass[11pt]{article}

\usepackage{amsmath}
\usepackage{amssymb}
\usepackage{latexsym}
\usepackage{graphicx}
\usepackage{amsthm}
\usepackage{amsbsy}
\usepackage{color}
\usepackage{bbm}
\usepackage{float}
\usepackage[T1]{fontenc}

\usepackage[a4paper,top=3cm,bottom=2cm,left=3cm,right=3cm,marginparwidth=1.75cm]{geometry}

\newcommand{\ignore}[1]{}

\newtheorem{theorem}{Theorem}[section]
\newtheorem{lemma}{Lemma}[section]

\title{Distributed Server Allocation for Content Delivery Networks}

\author{Sarath Pattathil, Vivek S.\ Borkar and  Gaurav S.\ Kasbekar		
\thanks{Sarath Pattathil, Vivek S.\ Borkar and Gaurav Kasbekar are with the Department of Electrical Engineering, Indian Institute of Technology Bombay, Powai, Mumbai 400076, India. Email: sarathpattathil@iitb.ac.in, borkar.vs@gmail.com, gskasbekar@ee.iitb.ac.in. The work of VSB was supported in part by a J.\ C.\ Bose Fellowship and a grant for `Approximation of High Dimensional Optimization and Control Problems' from the Department of Science and Technology, Government of India.}
        }

\date{}

\begin{document}

\maketitle

\begin{abstract}
We propose a dynamic formulation of file-sharing networks in terms of an average cost Markov decision process with constraints. By analyzing a Whittle-like relaxation thereof, we propose an index policy in the spirit of Whittle and compare it by simulations with other natural heuristics.
\end{abstract}




\section{Introduction}
Recently, Content Delivery Networks (CDNs), which distribute content (\emph{e.g.}, video and audio files, webpages) using a network of server clusters situated at multiple geographically distributed locations, have been extensively deployed in the Internet by content providers themselves (\emph{e.g.}, Google) as well as by third-party CDNs that distribute content on behalf of multiple content providers (\emph{e.g.}, Akamai's CDN distributes Netflix and Hulu content)~\cite{RF:Kurose:Ross},~\cite{RF:leighton:CDNs}. The delay incurred in downloading content to an end user is often significantly lower when a CDN is used compared to the case where all content is downloaded from a single centralized host, since the server clusters of a CDN are located close to end users~\cite{RF:Kurose:Ross},~\cite{RF:leighton:CDNs}.

In this paper, we consider a server cluster which contains $M \geq 2$ servers and is part of a CDN. The server cluster stores $N$ large file types (\emph{e.g.}, videos). There is a high demand for each file type and therefore each file type is replicated across multiple servers within the cluster. Each file type is characterized by the average size of the file it stores. We do not maintain the identity of each individual file for every file type, but instead assume that the size of each file from any particular file type comes from a distribution. From now on, we refer to the file types as  files for sake of brevity. Requests for the $N$ files from end users or from smaller server clusters arrive at the server cluster from time to time. There are two approaches to serving the file requests~\cite{Virag,Virag2}:
\begin{enumerate}
\item   \emph{Single Server Allocation}: Each file request is served by a single server~\cite{Virag,Virag2}.

\item  \emph{Resource Pooling}: Each file request is simultaneously served by multiple servers, in particular, different chunks of the file are served by different servers in parallel~\cite{Virag,Virag2}.

\end{enumerate}
Resource pooling has been found to outperform single server allocation in prior studies~\cite{Virag,Virag1,Virag2} and hence in this paper  we assume that resource pooling is used. Also, we allow multiple files to be simultaneously downloaded from a given server. At any time instant, the sum of the rates at which a server $j$ transmits different files is constrained to be at most $\mu_j$. Requests for different files are stored in different queues, and there is a cost for storing a request in a queue. Let $\xi^{ij}(t)$ be the rate at which server $j$ transmits file $i$ at time $t$. We consider the problem of determining the rates  $\xi^{ij}(t)$ for each $i$, $j$ and $t$ so as to minimize the total storage cost. We formulate this problem as a Markov Decision Process (MDP)~\cite{Hern}. We show that this problem is Whittle-like indexable~\cite{Whittle} and use this result to propose a Whittle-like scheme~\cite{Whittle} that can be implemented in a distributed manner\footnote{We use the phrase `Whittle-like' instead of just Whittle because the scheme introduced in this paper, although in the same spirit of Whittle's original paper, is not exactly the same.}. We evaluate the performance of our scheme using simulations and show that it outperforms several natural heuristics for the problem such as Balanced Fair Allocation, Uniform Allocation, Weighted Allocation, Random Allocation and Max-Weight Allocation.

We now review related prior literature. In \cite{Virag1}, performance of Content Delivery Networks is evaluated in  a static framework. This work also studies the  tradeoffs  between delay for each packet vs the  energy used etc. The polymatroid structure of the service capacity in this model is exploited to get an expression for mean file transfer delay that is experienced by incoming file requests. Performance of dynamic server allocation strategies, such as random server allocation or allocation of least loaded server, are also explored. We use the model  of \cite{Virag1} for CDN, but go a step further by looking at a fully dynamic optimization problem as an MDP.

In \cite{Virag2}, a centralized content delivery system with collocated servers is studied. Files are replicated in these servers and these serve as a pooled resource which cater to file requests. The article shows how dynamic server capacity allocation outperforms simple load balancing strategies such as those which assign the least loaded server, or assign the servers at random. The article also goes on to study file placement strategies that improve the utility of the system.

Several works including \cite{Laconte, Laconte2, Sharayu, Zhou} look at large-scale content delivery networks, focusing on placement of content in the servers. Of these, \cite{Laconte} also studies the greedy method of server allocation and its efficiency under various regimes of server storage capacities, and under what content placement strategy it would be efficient. Article \cite{Zhou} studies strategies for scheduling after the content placement stage, and proposes an algorithm, called the Fair Sharing with Bounded Degree (FSBD), for server allocation.

In \cite{Virag3}, multiclass queuing systems are studied with different arrival rates. The service rates are constrained to be in a symmetric polymatroid region. Large scale systems with a growing number of service classes are studied and several asymptotic results regarding fairness and mean delays are obtained.

Multi-server models are studied in \cite{JohnT} with each server connected to multiple file types and each file type stored in multiple servers, thereby creating a bipartite graph. This article focuses on the scaling regime where the number of servers goes to infinity. It is shown that even if the average degree $d_n << n :=$ the number of servers, an asymptotically vanishing queuing delay can be obtained. These results are based on a centralized scheduling strategy.

In \cite{Bonald}, multi-server queues are studied with an arbitrary compatibility graph between jobs and servers. The paper designs a scheduling algorithm which achieves balanced fair sharing of the servers. Several parameters are analyzed using this policy by drawing a parallel between the state of the system at any time to that of a Whittle network.

However, none of the above papers~\cite{Bonald, Laconte, Laconte2, Sharayu, Virag2, Virag3, JohnT,  Zhou} show Whittle indexability of the respective resource allocation problems they address. The work closest in spirit to ours is  \cite{Larranga}, which studies a Whittle indexability scheme for birth and death restless bandits. These  model server allocation to queues, but it does not study the case when there are multiple servers storing the same file types as is the case in general content delivery networks. In the present work we take an alternative approach which considers a dynamic optimization or control problem that can be interpreted as a problem of scheduling  restless bandits. We analyze it in the framework laid down by Whittle for deriving a heuristic index policy \cite{Whittle}. To the best of our knowledge, this paper is the first to show Whittle-like indexability of the server allocation problem in the setting of a CDN server cluster that uses resource pooling, with the objective of minimization of the total file request storage cost. The fact that this problem is Whittle-like indexable allows us to decouple the original average cost MDP, which is difficult to solve directly, into more tractable separate control problems for individual file types. The decoupling leads to an efficient algorithm based on computation of Whittle-like indices, which outperforms several natural heuristics for the problem. Our proof techniques broadly follow the general scheme of \cite{Agarwal}, albeit with some differences.

 The Whittle index heuristic has been successfully applied to various resource allocation problems including: crawling for ephemeral content \cite{ephemeral}, congestion control \cite{ABSCDC}, UAV routing \cite{Ny},  sensor scheduling \cite{Nino},  routing in clusters \cite{Nino2}, opportunistic scheduling \cite{Opportunistic}, inventory routing \cite{Glaze4}, cloud computing \cite{Cloud} etc. General applications to resource allocation problems can be found in \cite{Larranga}. Book length treatments of restless bandits can be found in \cite{Jacko} and \cite{Ruiz}.

The rest of the paper is structured as follows:
In section \ref{sec:model_prob}, we discuss our model  and formulate the problem as a Markov Decision Process (MDP). Section \ref{sec:struc_prop} shows various structural properties of the value function of the MDP formulated in section \ref{sec:model_prob}. In section \ref{sec:whittle}, we prove that the problem of server allocation in the resource pooling setting is in fact  indexable and provide a scheme to compute this index. Section \ref{sec:simulations} discusses other heuristics for server allocation and presents numerical comparisons of the proposed index policy with other heuristics. We conclude the paper with a brief discussion  in Section \ref{sec:conc}.

We conclude this section with a brief introduction to the Whittle index \cite{Whittle}. Let $X^i(t), \ t \geq 0, 1 \leq i \leq N$, be $N$ Markov chains, each with two modes of operations: active and passive, with associated transition kernels $p_1( \cdot | \cdot ), p_0( \cdot | \cdot )$ resp. Let $r_1^i(X^i(t)), r_0^i(X^i(t))$ be instantaneous rewards for the $i^{th}$ user in the respective modes with $r^i_1( \cdot ) \geq r^i_0( \cdot ).$
The goal is to schedule active/passive modes so as to maximize the total expected average reward
$$\lim_{T \rightarrow \infty} \frac{1}{T} \sum_{t=0}^{T-1}\sum_j E[r^j_{\nu^j(t)}(X^j(t))]$$
where $\nu^j(t) = 1$ if $j$th process is active at time $t$ and $0$ if not, under the constraint $\sum_j\nu^j(t) \leq M  \ \forall t$, i.e., at most $M$ processes are active at each time instant. This hard constraint makes the problem difficult to solve (see \cite{Papa}). So following Whittle, one relaxes the constraint to
$$\lim_{T \rightarrow \infty} \frac{1}{T} \sum_{t=0}^{T-1} \sum_j E[\nu^j(t)] \leq M.$$
 This makes it a problem with separable cost and constraints which, given the Lagrange multiplier $\lambda$, decouples into $N$ individual problems with reward for passivity changed to $\lambda + r_0(\cdot)$.  The problem is Whittle indexable if under optimal policy, the set of passive states increases \textit{monotonically} from empty set to full state space as $\lambda$ varies from $-\infty$ to $+\infty$. If so, the Whittle index for a given state can be defined as the value of $\lambda$ for which both modes (active and passive) are equally desirable. The index policy is then to compute these for the current state profile, sort them in decreasing order, and render active the top $M$ processes, the rest passive. The decoupling  implies $O(N)$ growth of state space as opposed to the original problem, for which it is exponential in $N$. Further, the processes are coupled only through a simple index policy. The latter is known to be asymptotically optimal as $N\uparrow\infty$ \cite{Weber}. However, no convenient general analytic bound on optimality gap seems available.

\section{Model and problem formulation}
\label{sec:model_prob}

Consider a server cluster that contains multiple servers, each of which stores one or more files. We represent this system using a  bipartite graph\footnote{Recall that a graph $G = (V,E)$ is said to be \emph{bipartite} if its node set $V$ can be partitioned into two sets $F$ and $S$ such that every edge in $E$ is between a node in $F$ and a node in $S$~\cite{RF:west:graph:theory}.} $G = (F \cup S;E)$ where $F$ is a set of $N$ files,
$S$ is a set of $M$ servers, $E$ is the set of edges, and each edge $e \in E$ connecting a file $i \in F$ and server $j \in S$ implies that a copy of file $i$ is replicated at server $j$ (see Figure~\ref{fig:Graphical_Model}).

\begin{figure}[h!]
\begin{center}
\includegraphics[angle=0,scale=0.6]{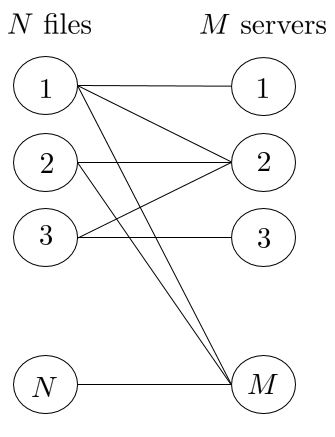}
\caption{The model used in this paper. A link between File $i$, and Server $j$ denotes that a copy of File $i$ is stored in Server $j$.}
\label{fig:Graphical_Model}
\end{center}
\end{figure}

For $j \in S$, $F_j$ denotes the set of files that are stored in server $j$. Similarly, for $i \in F$, $S_i$ denotes the set of servers that store file $i$.
Requests for file $i \in F$ arrive to the server cluster according to an independent Poisson process with rate $\Lambda^i$ and are queued in a separate queue for each file type. We assume that the job (requested file) sizes have an exponential distribution. (For sake of simplicity, we assume their means to be identically equal to one. More general cases can be handled by suitable scaling of the $\xi^{ij}(\cdot)$'s defined below.)
Let $\xi^{ij}(t)$ denote the rate at which server $j$ transmits file type $i$ at time $t$. Then the capacity constraint at each server can be expressed as:
\begin{align}
\label{eq:constraint}
\sum_{i \in F_j} \xi^{ij}(t) \ \le \ \mu_j \qquad \forall t \geq 0, \ j \in \{ 1,2,\cdots, M\},
\end{align}
where $\mu_j$ is the maximum permissible rate of transmission from server $j$.

Let $f^i(x)$ be the  cost for storing $x$ jobs in the queue $i$. We assume $f^i(\cdot)$ to be an increasing strictly convex function (see \cite{Bertsekas}) for $i = 1,2,\cdots, N$. (We comment on the strict convexity assumption at the end of Section \ref{sec:struc_prop}). Our aim is to minimize the long run average cost, given by
\begin{align}
\limsup_{T \rightarrow \infty}  \frac{1}{T} \int_{0}^{T} \mathbb{E}[\: \sum_{i=1}^{N} f^i(X^i(t))]dt, \nonumber
\end{align}
where $X^i(t)$ is the length of queue $i$ at time $t$.

This makes it a \textit{continuous time Markov decision process} with the state process given by $\hat{X}(t) = [X^1(t), \cdots, X^N(t)], t \geq 0$, taking values in the state space $\textbf{S}^N$ where $\textbf{S} := \{0,1,2,\cdots\}$ with control process $\xi(t) := \{\xi^{ij}(t)\}_{i \in F, j \in S_i}, t \geq 0$, taking values in the compact control space
$\textbf{U} := \{u^{ij}, i \in F, j \in S_i : \sum_{i \in F_j}u^{ij} \leq \mu_j \ \forall j\}.$ We shall consider as admissible control policies the $\{\xi^{ij}(\cdot)\}$ whereby one has the controlled Markov property, i.e., for $t \geq 0, \delta > 0$,
$$P(\hat{X}(t + \delta) = y | \hat{X}(s), \xi(s), s \leq t) = \hat{q}(y|x(t), \xi(t))\delta + o(\delta)$$
for a `controlled rate matrix' $q = [[q(y | x, u)]], x, y \in \textbf{S}^N, u \in \textbf{U}$. A special case is that of the control policies wherein $\xi(t)$ is adapted to $\hat{X}(s), s \leq t$, for all $t \geq 0$.  As usual, one has the important special subclasses of control policies, viz., stationary deterministic policy wherein $\xi(t) = v(\hat{X}(t))$ for a prescribed $v(\cdot) : \textbf{S}^N \mapsto \textbf{U}$, and stationary randomized policy wherein the conditional law of $\xi(t)$ given $\hat{X}(s), s \leq t,$, depends on $\hat{X}(t)$ alone.\\

\noindent \textbf{Stability Assumption:} We assume there exists a stationary randomized policy under which the cost is finite (which in particular implies that the policy is stable in the sense that the corresponding Markov chain $\hat{X}(\cdot)$ is positive recurrent), and, in addition,
\begin{equation}
\sum_{j \in S_i}\mu_j > \Lambda_i \ \forall i. \label{stability-local}
\end{equation}

The stability assumption above ensures the existence of at least one stationary randomized policy under which the process is stable. Our assumption on the $f^i$'s  implies that $\lim_{x \uparrow\infty}f^i(x) = \infty \ \forall i$, implying in turn that the cost is \textit{near monotone} \cite{Borkarsurvey} in the sense that it penalizes high values of the state $\|\hat{X}(t)\|$. In particular, an unstable control policy that leads to transience or null recurrence will lead to  an infinite cost. In fact it is known that for the problem without the additional constraint (\ref{eq:constraint}), an optimal stable stationary deterministic policy that is optimal among all admissible policies, exists under this condition. (See \cite{Borkarsurvey} for the discrete time case, the continuous time situation is completely analogous.) This remains true for the constrained problem we consider next after relaxing (\ref{eq:constraint}) to a weaker `average' constraint in the spirit of \cite{Whittle} below, if we replace `stationary deterministic policy' by `stationary randomized policy' in the above (again, see \cite{Borkarsurvey} for the discrete time case, the continuous time situation being completely analogous.) We shall not need the latter result, because the Whittle-like index policy we propose is in fact a stationary deterministic policy \textit{sans} any randomization.

Note that $X^i(\cdot), 1 \leq i \leq N$, are in fact individual controlled Markov chains coupled through their controls that have to satisfy the $M$ constraints (\ref{eq:constraint}) that couples them. This forces us to view the combined process $\hat{X}(\cdot)$ as a single controlled Markov chain. The Whittle device we use below allows us to undo this for purposes of analysis via a clever heuristic. Specifically, the controlled rate matrix $Q^i(t), t \geq 0,$ of $X^i(\cdot)$ is  given by: for $z > 0$,
\begin{align*}
Q^i(z+1|z,\xi^{ij}(t), \:  j \in S_i) &= \Lambda^i, \\
Q^i(z-1|z,\xi^{ij}(t), \:  j \in S_i) &= \sum_{j \in S_i}\xi^{ij}(t),\\
Q^i(z|z,\xi^{ij}(t), \:  j \in S_i) &= - \big( \Lambda^i + \sum_{j \in S_i}\xi^{ij}(t) \big),
\end{align*}
for $z = 0$,
\begin{align*}
Q^i(1|0,\xi^{ij}(t), \:  j \in S_i) &= \Lambda^i, \\
Q^i(0|0,\xi^{ij}(t), \:  j \in S_i) &= -\Lambda^i.
\end{align*}
Following the classic paper of Whittle \cite{Whittle}, we relax the $M$ per stage constraints (\ref{eq:constraint}) to $M$ averaged constraint
\begin{align}
\label{eq:relaxed_constraint}
\limsup_{T \rightarrow \infty}  \frac{1}{T} \int_{0}^{T} \sum_{i \in F_j} \mathbb{E}[ \xi^{ij}(t)] \: \leq \mu_j \quad \forall  j \in \{ 1,2,\cdots, M\}
\end{align}
where we assume that $0 \leq \xi^{ij}(t) \leq \mu_j \ \forall \ i,j,t$. Specifically, we have replaced the $M$ hard constraints (\ref{eq:constraint}) that apply at \textit{each} time instant by $M$ \textit{average} constraints which allow the violation of (\ref{eq:constraint}) from time to time, but requires it to hold only in an average sense. In particular, the left hand side of (\ref{eq:relaxed_constraint}) can be viewed as another average cost functional. This makes it a classical  constrained Markov decision process \cite{Borkarsurvey}. This has an equivalent formulation as a linear program on the space of measures, in terms of the so called \textit{ergodic occupation measures} \cite{Borkarsurvey}. These measures are defined as probability measures on the product space \textbf{$S^N\times U$} that are of the form
$$\Phi(dx, du) = \Phi_0(dx)\Phi_1(du | x)$$
where $\Phi_0$ is the marginal on \textbf{$S^N$} which is required to be the stationary distribution of the Markov chain controlled by $\Phi_1(du | x)$, the regular conditional law in the above decomposition interpreted as a stationary randomized policy. The control problem can then be identified with the problem of minimizing the integral of the running cost $\hat{f}(\cdot, \cdots, \cdot) := \sum_if^i(\cdot)$ w.r.t.\ this measure, a linear functional thereof, over the set of all ergodic occupation measures which turns out to be a closed convex set characterized by a set of linear equalities and inequalities. Specifically, one has:\\

\noindent Minimize $\int \hat{f}d\Phi(dx, \textbf{U})$\\

\noindent subject to:  $\Phi \geq 0, \ \Phi(\textbf{S}^N\times\textbf{U}) = 1,$\\

\noindent $\int\Phi(dx,du)\prod_iQ^i(y^i | x^i, u^{ij}, j \in S_i) = 0.$\\

See \cite{Borkarsurvey} for details. This facilitates the use of standard tools of abstract convex optimization in this context. While we do not need the details thereof here, we do require one consequence of it, viz., that it allows one to consider an equivalent   unconstrained  average cost problem with cost
\begin{align}
 \lim_{T \uparrow \infty} \frac{1}{T} \int_{0}^{T} \mathbb{E}\left[ \: \sum_{i=1}^{N} f^i(X^i(t)) + \sum_{j} \hat{\lambda}_j (\sum_{i \in F_j} \xi^{ij}(t) - \mu_j )\right]dt, \nonumber
\end{align}
where $\hat{\lambda}_j \geq 0$ is the Lagrange multiplier associated with the $j^{th}$ relaxed constraint
$$ \lim_{T \uparrow \infty}  \frac{1}{T} \int_{0}^{T} \mathbb{E}[\sum_{i} \xi^{ij}(t)] \leq \mu_j. $$
(We replace the conventional `$limsup_{T\uparrow\infty}$' in analysis of average cost control by `$lim_{T\uparrow\infty}$' by exploiting the fact that the results of \cite{Borkarsurvey} allow us to restrict to stationary randomized policies for which the $limsup_{T\uparrow\infty}$ above is in fact the $lim_{T\uparrow\infty}$.) Since the cost is now separable in $X^i(\cdot)$'s, given the values of the Lagrange multipliers $\hat{\lambda}_j$, this optimization problem decouples into separate control problems for individual processes $X^i(\cdot)$, with the cost function for the $i$th process (file type) being given by:
\begin{align}
c^i(x^i, \hat{\lambda}) = f^i(x^i) + \sum_{j \in S_i}\hat{\lambda}_j (\xi^{ij}(t) - \mu_j ). \nonumber
\end{align}
where $\hat{\lambda} = [\hat{\lambda}_1, \hat{\lambda}_2,  \cdots \hat{\lambda}_M]$ is a vector containing all $\hat{\lambda}_j's$. The average cost dynamic programming (DP) equation for this  MDP for file type $i$ is given by \cite{Hern}:
\begin{align}
\min_{\mu_j \geq u^{ij} \geq 0, j \in S_i}  \bigg( c^i(x, & \hat{\lambda})  - \beta^i \nonumber \\
&+ \sum_y V_{\hat{\lambda}}^i(y) Q^i(y|x,{\xi}^{ij}, \:  j \in S_i) \bigg) = 0
\end{align}
where:
\begin{itemize}

\item  $\beta^i$ is the optimal cost for file type $i$,

\item $V_{\hat{\lambda}}^i(\cdot)$ is the value function (sometimes called the `relative value function').
\end{itemize}
In what follows, we drop the dependence of $V_{\hat{\lambda}}^i(\cdot)$ on $i$ and $\hat{\lambda}$ for sake of notational simplicity and bring it back only when needed for the analysis. Substituting the values of $Q^i$ back in the DP equation and dropping the superscript $i$ (except from ${\xi}^{ij}$) for ease of notation, we have\footnote{Note that when the queue of file type $i$ is empty, no server needs to provide any service to that particular file type.}: for $x > 0$,
\begin{align}
\min_{\mu_j \geq {\xi}^{ij} \geq 0, j \in S_i} \bigg( c(x, \hat{\lambda}) - \beta + V & (x+ 1) \Lambda  + V(x-1) \sum_{j \in S_i}{\xi}^{ij} \nonumber \\
&- V(x)\big( \Lambda + \sum_{j \in S_i}{\xi}^{ij} \big)  \bigg) = 0,
\end{align}
equivalently,
\begin{align}
\label{eq:CTMC_MDP}
 & \min_{\mu_j \geq u^{ij}  \geq 0, j \in S_i}  \bigg( f(x) + \sum_{j \in S_i}\hat{\lambda}_j ({\xi}^{ij} - \mu_j ) - \beta + V(x+1) \Lambda \nonumber \\
&+ V(x-1) \sum_{j \in S_i}{\xi}^{ij}(t) - V(x)\big( \Lambda + \sum_{j \in S_i}{\xi}^{ij} \big)  \bigg) = 0.
\end{align}
Adding  $V(x)$ on both sides of equation (\ref{eq:CTMC_MDP}) we get:
\begin{eqnarray}
V(x) &=&  \min_{\mu_j \geq u^{ij} \geq 0, j \in S_i}  \bigg( f(x) + \sum_{j \in S_i}\hat{\lambda}_j ({\xi}^{ij} - \mu_j) - \beta \nonumber \\
&& \qquad + \ V(x+1) \Lambda^i
+ V(x-1) \sum_{j \in S_i}{\xi}^{ij}  \nonumber \\
&& \qquad \qquad + \ V(x)\big(1 - \big( \Lambda^i + \sum_{j \in S_i}{\xi}^{ij} \big) \big)  \bigg) \label{eq:Original_DP}
\end{eqnarray}
The equations for $x = 0$ can be written in a simiar fashion with appropriate modifications.

We now adapt the idea of uniformization to pass from a continuous time Markov chain to a discrete time Markov chain. If we scale all transition rates by a fixed multiplicative factor, it is tantamount to time scaling which will scale the average cost, but not affect the optimal policy. Hence without loss of generality, we can assume that the arrival and service rates are such that the coefficients of $V(\cdot)$ that appear in the right hand side of equation (\ref{eq:Original_DP}) are between some $\epsilon > 0$ and $1$ and  can be interpreted as transition probabilities of a discrete time controlled Markov chain. Thus (\ref{eq:Original_DP}) is a dynamic programming equation for a discrete time Markov decision process with average cost. Note that the equation at best specifies $V$ only up to an additive scalar, so for its well-posedness, in the least we need to add a qualification such as (say) $V(0) = 0$. We shall make this choice (which is by no means unique) and stay with it. See \cite{Borkarbook}, Chapter VI, (in particular, Theorem 4.1, p.\ 87) for a complete treatment of well-posedness of (\ref{eq:Original_DP}).  One only needs to verify the assumption therein of `stability under local perturbation' which states that a stable stationary deterministic policy remains so if we change the control choice at exactly one state. This is immediate if each state has at most finitely many successors, as is the case here - see Lemma 1.1, p.\ 71, of \cite{Borkarbook}. We take the foregoing as given, suffice to say that the near-monotonicity of the cost and existence of a stable stationary randomized policy with finite cost  by virtue of the `Stability Assumption' above play a crucial role in establishing the DP equation.

As we are working with a fixed $i$, the control space is $U^i := \prod_{j \in S_i}[0, \mu_j]$ and a stationary deterministic policy corresponds to $\xi^{ij}(n) = \varphi(X(n))$ for a measurable $\varphi: \textbf{S} \mapsto U^i$, where $X(\cdot)$ is the corresponding controlled Markov chain, now in discrete time (We drop the superscript $i$ for notational convenience.). We shall identify this policy with the map $\varphi$ by a standard abuse of notation.

The  expression which is to be minimized on the right hand side of (\ref{eq:Original_DP})  is linear in ${\xi}^{ij}, j \in S_i$ and each ${\xi}^{ij}$ has the capacity constraint which restricts the values of ${\xi}^{ij}$ to be $\leq \mu_j$, i.e., ${\xi}^{ij} \in [0, \mu_j]$. This, combined with the fact that the objective is linear, ensures that the minimum is attained at a corner where each server is either serving at full capacity or at zero capacity, i.e., at ${\xi}^{ij} = 0$ or ${\xi}^{ij} = \mu_j$ for all $j \in S_i$. Define $u_{ij} = 1$ if ${\xi}^{ij} = \mu_j$ and $u_{ij} = 0$ if ${\xi}^{ij} = 0$.

This achieves the first simplification in Whittle's program, viz., to decouple the original hard problem into $N$ simpler problems. But unlike in the original Whittle case, where the decision was binary between active and passive modes, we have multiple decision variables, ${\xi}^{ij}$ for each $i$. The foregoing shows that each one separately entails a binary decision between $0$ and $\mu_j$ resp. Our approach to arriving at a Whittle-like policy is the most common one, viz., to first show the existence of an optimal threshold policy and then establish the monotonicity of the threshold in the Lagrange multiplier. Even the notion of a threshold does not make sense in a control space without a natural order, thus we need to reduce the problem to a situation where such is the case. This suggests that we apply the Whittle philosophy separately to each control variable in isolation, keeping the rest fixed at their respective capacities $\mu_{[\cdot]}$. We make this the basis for coming up with a Whittle-like index policy. Like the original Whittle scheme, this too is a heuristic, which we later compare with other natural heuristics empirically and find that it performs quite well in comparison. Our motivation for this specific choice and no other is as follows. In principle, we could fix any values of all but one control variable in order to reduce it to a single control variable case, but fixing the rest at maximum rate, which aids stability, puts the least onus on the flagged control variable vis-a-vis stability. To amplify this point, consider, e.g., the other extreme where we fix all other rates to zero. Then to ensure the existence of at least one stable stationary randomized policy for the decoupled problem, we would need a stronger restriction than the above `Stability Assumption'. Observe in particular that we are now considering separate control problems associated not only for each  process $i$ separately, but for separate \textit{pairs} of process $X^i(\cdot)$ and control $\xi^{ik}(\cdot)$ for a prescribed $k$, having fixed $\xi^{ij}(\cdot) \equiv \mu_j \ \forall \ j \neq k$. The sole variable being manipulated now takes values in an ordered set $[0, \mu_k]$ which facilitates search for an optimal threshold policy.

This also has the added bonus that all but one Lagrange multiplier drop out of each such DP equation, facilitating later the definition of Whittle-like index that would otherwise be quite messy.

We emphasize again that this is a heuristic policy just like the original Whittle case and need not be optimal. An optimal policy for the exact coupled problem will face the curse of dimansionality in a major way. To see this, suppose we use finite buffers of a constant size for each queue as an approximation and assume $|S_i|, |F_j|$ are independent of $i,j$ resp., denoted simply as $|S|, |F|$ resp.  The state space for the original problem is the product of individual state spaces of the queues, which grows exponentially in $|S|$. In contrast, after decoupling the problem using Lagrange multipliers, it grows linearly in $|S|$.  This is exactly the same problem which motivates the original Whittle index.

Since all other servers are serving at full rate, we have that
\begin{align}
\xi^{ij}(t) = \mu_j \qquad \forall \: j \in S_i, j \neq k, \ \forall t. \text{ s.t.  } X^i(t) > 0. \nonumber
\end{align}
Let $\lambda_k = -\hat{\lambda}_k \mu_k$. We interpret  $\lambda_k$ as the marginal disutility of allowing server $k$ to serve at $0$ rate when all other servers containing the file type are already serving at their full capacity. This disutility plays the role of `subsidy' in the original Whittle formulation which dealt with a reward maximization problem instead of cost minimization. On substituting ${\xi}^{ij} = \mu_j  \  \forall \: j \in S_i, j \neq k$, we have:
\begin{align}
\label{eq:Simple_DP}
V(x) = f(x) - \beta + \min \bigg( & \lambda_k + \sum p_1(y|x)V(y), \nonumber \\
& \qquad \qquad \quad \sum p_2(y|x)V(y)  \bigg).
\end{align}
Here $p_1( \cdot | \cdot )$ is the transition probability when the server does not serve this file type and is given by (for $x > 0$):
\begin{align}
\label{eq:def_passive_prob}
p_1(x+1|x) &= \Lambda, \nonumber \\
p_1(x|x) &= 1 - \big( \Lambda + \sum_{j \in S_i, j \neq k} \mu_j \big), \nonumber \\
p_1(x-1|x) &= \sum_{j \in S_i, j \neq k} \mu_j,
\end{align}
and $p_2( \cdot | \cdot )$ is the transition probability when the server serves this file type and is given by (for $x > 0$):
\begin{align}
\label{eq:def_active_prob}
p_2(x+1|x) &= \Lambda, \nonumber \\
p_2(x|x) &= 1 - \big( \Lambda + \sum_{j \in S_i} \mu_j \big), \nonumber \\
p_2(x-1|x) &= \sum_{j \in S_i} \mu_j.
\end{align}
For $x=0$, the transition probabilities $p_1( \cdot | \cdot )$ and $p_2( \cdot | \cdot )$ are the same and are given by (for $i = 1,2$):
\begin{align}
p_i(1|0) &= \Lambda \nonumber \\
p_i(0|0) &= 1 - \Lambda \nonumber
\end{align}

In the next section, we prove some structural properties of the value function.

\section{Structural Properties of the Value function}
\label{sec:struc_prop}

This section closely follows in spirit the approach of  \cite{Agarwal}, \cite{Borkarbook}, \cite{Cloud} and \cite{proc_sharing}, but with significantly different proofs.

\begin{lemma}
$V(\cdot)$ is non-decreasing in the number of files.
\label{lemma_inc_file}
\end{lemma}
\begin{proof} (Sketch)
We use a `pathwise coupling' argument. Consider initial conditions $x < x'$ in $\textbf{S}$ and an optimal, therefore stable (i.e., positive recurrent) stationary deterministic policy $v(\cdot)$. Consider the controlled chains $X(n), X'(n), n \geq 0,$ as follows: We use the standard formulation of a controlled Markov chain as a dynamics driven by control and noise, i.e.,
\begin{eqnarray*}
X(n+1) &=& F(X(n), \xi(n), \zeta(n+1)), \\
X'(n+1) &=& F(X'(n), \xi(n), \zeta(n+1)),
\end{eqnarray*}
with $X(0) = x, X'(0) = x'$, where $\{\xi_n\}$ is the control process, $\{\zeta(n)\}$ is i.i.d.\ noise uniform on $[0,1]$, and $F$ is some measurable map. Note that the map $F$, the driving noise $\{\zeta(n)\}$, and the control sequence $\{\xi(n)\}$ is common across both. It is always possible to replicate the processes in law on a common probability space in this fashion. In addition, we choose $\xi(n) = v(X'(n)) \ \forall n$. This choice is optimal for $X'(\cdot)$, but not for $X(\cdot)$. In particular, $X'(\cdot)$ is a positive recurrent Markov chain and hits state $0$ infinitely often with probability $1$. Each time this happens, $X'(\cdot) - X(\cdot)$ drops by $1$, hence
$$\tau := \min\{n \geq 0 : X'(n) = X(n)\} < \infty \ \mbox{a.s.}.$$
Note that by our construction,
\begin{itemize}

\item we have:
\begin{eqnarray}
X'(m) &>& X(m) \ \forall \ m < \tau, \label{one'} \\
&=& X(m) \ \mbox{for} \ m \geq \tau, \label{two'} 
\end{eqnarray}
and,

\item for $n < \tau$, either $X'(m+1) - X(m+1) = X'(m) - X(m)$ or $X'(m+1) - X(m+1) = X'(m) - X(m) - 1$ and the latter case occurs only if $X(m) = X(m + 1) = 0$ and $X'(m+1) = X'(m) - 1$.
    \end{itemize}
For $x = X(m)$, resp., $X'(m)$, (\ref{eq:Original_DP}) leads to
$$E\left[V(X'(m)) - V(X(m))\right] \geq E\left[V(X'(m+1)) - V(X(m+1))\right]$$
Iterating, we get for $T \geq 1$,
$$V(x') - V(x) \geq E\left[V(X'(\tau\wedge T)) - V(X(\tau\wedge T))\right].$$
Letting $T\uparrow \infty$ and using Fatou's lemma, we have
$$V(x') - V(x) \geq E\left[V(X'(\tau)) - V(X(\tau))\right] = 0.$$
\end{proof}

\begin{lemma}
\label{lemma:inc_diff}
$V(\cdot)$ is strictly convex, strictly increasing, and has the property of increasing differences, i.e., for $z>0$ and $x>y$.
\begin{align}
V(x+z) - V(x) > V(y+z) - V(y). \nonumber
\end{align}
\end{lemma}
\begin{proof}

The proof follows along similar lines as Lemma 6 in \cite{proc_sharing} and Theorem 4 in \cite{Agarwal}, but with several crucial differences. The argument uses induction. 
We embed the state space to the positive real line, $\mathbb{R}^+$.
Take $x_1, x_2 \in \textbf{S}, x_2 > x_1 > 0$. Let $V_n(\cdot)$ denote the $\alpha-$discounted $n-$step problem 
(For $x < 0$, we define $V_n(x) = V_n(0)$).
Let $u$ be the optimal control for state $x$ at time $n$. We have
\begin{align}
{V}_{n}  (x) &= f(x)  + \alpha V_{n-1} (x+1) \Lambda + \alpha V_{n-1}(x) \Big(1 - \Lambda  \nonumber \\
& \quad - \sum_i \mu_i\Big) +   \alpha V_{n-1}(x-1) \sum_{j\neq k} \mu_j + (1 - u) \lambda_k  \nonumber \\
& \quad \ + \  \alpha (1-u) \mu_k V_{n-1}(x) + \alpha u \mu_k V_{n-1}(x-1). \label{valueiteration1}
\end{align}

\noindent We have that $V_0(x) \equiv f(x),$ which is strictly convex. Assume that $V_{n-1}$ is convex. For $x_1, x_2$ as above, let $u_i, i = 1,2$, be the minimizers for $x = x_i, i = 1,2,$ resp.\ in (\ref{valueiteration1}).  Then
\begin{align}
 {V}_n  & (x_1)  + {V}_n(x_2) = f(x_1) + f(x_2)  \nonumber \\
&+  \alpha V_{n-1}(x_1) (1 - \Lambda - \sum_j \mu_j) + \alpha V_{n-1}(x_2) \bigg(1 - \Lambda - \sum_j \mu_j\bigg) \nonumber \\
& +  \alpha V_{n-1} (x_1+1) \Lambda + \alpha V_{n-1} (x_2+1) \Lambda \nonumber \\
& +  \alpha V_{n-1}(x_1-1) \sum_{j \neq k} \mu_j +  \alpha V_{n-1}(x_2-1) \sum_{j \neq k} \mu_j \nonumber \\
&+ (1 - u_1) \lambda_k + \alpha (1-u_1) \mu_k V_{n-1}(x_1) + \alpha u_1 \mu_k V_{n-1}(x_1-1) \nonumber \\
&+ (1 - u_2) \lambda_k + \alpha (1-u_2)\mu_k V_{n-1}(x_2) + \alpha u_2  \mu_k V_{n-1}(x_2-1). \nonumber
\end{align}
Consider two separate cases depending on the  values of $u_1, u_2$. \\ \vspace{-2mm} \\
\textbf{Case 1}: $u_1 = u_2$
\begin{align}
{V}_n &  (x_1) + {V}_n(x_2) \nonumber \\
&\geq^{*1} 2f\left(\frac{x_1 + x_2}{2}\right)    + 2\alpha V_{n-1}\left(\frac{x_1 + x_2}{2}\right) (1 - \Lambda - \sum_j \mu_j) \nonumber \\
& \quad + 2\alpha V_{n-1} \left(\frac{x_1 + x_2}{2}+1\right) \Lambda +  2\alpha V_{n-1}\left(\frac{x_1 + x_2}{2}-1\right) \sum_{j \neq k} \mu_j \nonumber \\
&\quad + 2\lambda_k\left(1 - \frac{u_1+u_2}{2}\right) + 2 \alpha \left(1 - \frac{u_1+u_2}{2}\right)\mu_k V_{n-1}\left(\frac{x_1 + x_2}{2} \right) \nonumber \\
&\quad + 2 \alpha \left(\frac{u_1+u_2}{2}\right)\mu_k V_{n-1}\left(\frac{x_1 + x_2}{2} -1 \right) \nonumber 
\end{align}
 \vspace{-5.3mm}
\begin{align}
&\geq^{*2} 2f\left(\frac{x_1 + x_2}{2}\right) + 2 \alpha V_{n-1}\left(\frac{x_1 + x_2}{2}\right) (1 - \Lambda - \sum_j \mu_j) \nonumber \\
&\quad + 2 \alpha V_{n-1} \left(\frac{x_1 + x_2}{2}+1\right) \Lambda +  2 \alpha V_{n-1}\left(\frac{x_1 + x_2}{2}-1\right) \sum_{j \neq k} \mu_j \nonumber \\
&\quad + 2\lambda_k(1 - u_3) + 2 \alpha (1 - u_3)\mu_k V_{n-1}\left(\frac{x_1 + x_2}{2} \right) \nonumber \\
&\quad + 2 \alpha u_3\mu_k V_{n-1}\left(\frac{x_1 + x_2}{2} -1 \right) \nonumber \\
&= 2 {V}_n \left(\frac{x_1 + x_2}{2}\right). \nonumber
\end{align}
Here $u_3$ is the optimal control when the state is $\frac{x_1+x_2}{2}$. Inequality $*1$ follows from the convexity of $f(\cdot)$ and $V_{n-1}(\cdot)$. Inequality $*2$ follows from the definition of the optimal control $u_3$. \\ \vspace{-2mm} \\
\textbf{Case 2:} $u_1 \neq u_2$: \\
Consider the case $u_2=0, u_1 = 1$ (The other case is similar)
\begin{align}
{V}_n & (x_1) + {V}_n(x_2)  \nonumber \\
&\geq^{*1} 2f\left(\frac{x_1 + x_2}{2}\right) + 2 \alpha V_{n-1}\left(\frac{x_1 + x_2}{2}\right) (1 - \Lambda - \sum_j \mu_j)  \nonumber \\
&\quad + 2 \alpha V_{n-1} \left(\frac{x_1 + x_2}{2}+1\right) \Lambda + 2 \alpha V_{n-1}\left(\frac{x_1 + x_2}{2}-1\right) \times  \nonumber \\
&\qquad \sum_{j \neq k} \mu_j + \lambda_k + \alpha \mu_k V_{n-1}(x_2) + \alpha \mu_k V_{n-1}(x_1-1) \nonumber
\end{align}
\begin{align}
&= 2f\left(\frac{x_1 + x_2}{2}\right) + 2\lambda_k  + 2 \alpha V_{n-1}\left(\frac{x_1 + x_2}{2}\right) (1 - \Lambda - \sum_i \mu_i) \nonumber \\
&\quad + 2 \alpha V_{n-1} \left(\frac{x_1 + x_2}{2}+1\right) \Lambda +   2 \alpha V_{n-1}\left(\frac{x_1 + x_2}{2}-1\right) \sum_{i \neq k} \mu_i \nonumber \\
&\quad + 2\lambda_k\left(1-\frac{1}{2}\right) + 2\alpha \mu_k \left[ \frac{1}{2}V_{n-1}(x_2) + \frac{1}{2}V_{n-1}(x_1 - 1) \right] \nonumber 
\end{align}
\begin{align}
& \geq^{*2} 2f\left(\frac{x_1 + x_2}{2}\right) + 2\alpha V_{n-1}\left(\frac{x_1 + x_2}{2}\right) (1 - \Lambda - \sum_i \mu_i) \nonumber \\
&\quad + 2 \alpha V_{n-1} \left(\frac{x_1 + x_2}{2}+1\right) \Lambda + 2 \alpha V_{n-1}\left(\frac{x_1 + x_2}{2}-1\right) \sum_{i \neq k} \mu_i \nonumber \\
&\qquad \qquad + 2\lambda_k\left(1 - \frac{1}{2}\right) + 2 \alpha \mu_k \bigg[ \frac{1}{2} V_{n-1}\left(\frac{x_1 + x_2}{2}\right) \nonumber \\
&\qquad \qquad \qquad \qquad + \frac{1}{2} V_{n-1}\left(\frac{x_1 + x_2}{2} -1\right) \bigg] \nonumber 
\end{align}
\begin{align}
&\geq^{*3} 2f\left(\frac{x_1 + x_2}{2}\right) +  2 \alpha V_{n-1}\left(\frac{x_1 + x_2}{2}\right) (1 - \Lambda - \sum_i \mu_i) \nonumber \\
&\quad + 2 \alpha V_{n-1} \left(\frac{x_1 + x_2}{2}+1\right) \Lambda + 2 \alpha V_{n-1}\left(\frac{x_1 + x_2}{2}-1\right) \sum_{i \neq k} \mu_i \nonumber \\
&\qquad \qquad + 2\lambda_k(1 - u_3) + 2 \alpha (1 - u_3) \mu_k V_{n-1}\left(\frac{x_1 + x_2}{2} \right) \nonumber \\
&\quad \qquad \qquad \qquad + 2 \alpha u_3 \mu_k V_{n-1}\left(\frac{x_1 + x_2}{2} -1 \right) \nonumber \\
&= 2 {V}_n \left(\frac{x_1 + x_2}{2}\right). \nonumber
\end{align}
Here $u_3$ is the optimal control when the state has $\frac{x_1+x_2}{2}$ files. Inequalities $*1, *2$ follow from the convexity of $f(\cdot)$ and $V_{n-1}(\cdot)$ (we use the fact that convexity implies non-decreasing differences, i.e., $f(x + a) - f(x) \geq f(y + a) - f(y)$ for $x > y, \ a > 0$). Inequality $*3$ follows from the definition of the optimal control $u_3$.\\
Next consider the case where $x_1 > x_2=0$. We have:
\begin{align}
{V}_n(0) = f(0) + (1-u)\lambda_k + \alpha (1-\Lambda)V_{n-1}(0) + \alpha \Lambda V_{n-1}(1) \nonumber
\end{align}
From this equation, we see that $u=1$ if $\lambda_k > 0$ and $u=0$ otherwise. We rearrange the above equation as 
\begin{align}
{V}_n(0) &= f(0) + (1-u)\lambda_k + \alpha (1-\Lambda - \sum_{i}\mu_i)V_{n-1}(0) \nonumber \\
& \qquad \quad + \alpha \sum_{i \neq k}\mu_iV_{n-1}(0) + \alpha \mu_kV_{n-1}(0) + \alpha \Lambda V_{n-1}(1). \nonumber
\end{align}
We have:
\begin{align}
& {V}_n  (x_1) + {V}_n(0) \nonumber \\
& = f(x_1) + f(0)  + \alpha V_{n-1}(x_1) \left(1 - \Lambda - \sum_i \mu_i\right)  \nonumber \\
& + \alpha V_{n-1}(0) \left(1 - \Lambda - \sum_i \mu_i\right)  + \alpha V_{n-1} (x_1+1) \Lambda +  \alpha V_{n-1} (1) \Lambda \nonumber \\
& \qquad +  \alpha V_{n-1}(x_1-1) \sum_{i \neq k} \mu_i + \alpha V_{n-1}(0) \sum_{i \neq k} \mu_i \nonumber \\
& \qquad \quad + (1 - u_1) \lambda_k + (1-u_1)\alpha \mu_k V_{n-1}(x_1)  \nonumber \\
& \qquad \quad \quad+ \alpha u_1  \mu_k V_{n-1}(x_1-1) + (1 - u_2) \lambda_k + \alpha \mu_k V_{n-1}(0)  \nonumber \\
&\geq^{*1} 2{V}_n\left(\frac{x_1}{2}\right), \nonumber
\end{align}
where ${*1}$ is derived using convexity and by following similar arguments as in the case when $x_2 > 0$

Therefore, by induction, we have that $V_n$ is convex for all $n$. From equation \eqref{valueiteration1}, we see that $V_n$ is the sum of a strictly convex function $f$ and a convex function $V_{n-1}$ when $x \geq 0$. This shows that $V_n$ is in fact a strictly convex function for $x> 0$.   (Note that $V_0 = f$, which is also strictly convex). Letting $\tilde{V}_{\alpha}$ denote the value function of the infinite horizon $\alpha$-discounted problem, we have $V_n \to \tilde{V}_{\alpha}$ pointwise by convergence of the value iteration algorithm. Since $V_n(x) - f(x), \: x \geq 0$ is  convex for all $n$ and convexity is preserved under pointwise convergence,  $\tilde{V}_{\alpha}(x) - f(x), x \geq 0,$ is convex for all $\alpha$. Letting $\bar{V}_{\alpha}(x) := \tilde{V}_{\alpha}(x) - \tilde{V}_{\alpha}(0)$, so will be $\bar{V}_{\alpha} - f$ for all $\alpha$. By the vanishing discount argument of \cite{Agarwal}, $\bar{V}_{\alpha} \to$ the value function $V$ of the average cost problem, pointwise. Thus $V - f$ is convex. Since $f$ is strictly convex, it follows that $V(x), \ x \geq 0,$ is strictly  convex. Strict convexity and non-decreasing property imply strict increase on $[0, \infty)$. Strict convexity also implies  strictly increasing differences. This proves the claim.

\end{proof}

\begin{lemma}
\label{lemma:threshold_policy}
The optimal policy is a threshold policy, i.e.,  $\exists \ x^*$ such that if $x>x^*$, the server serves at full capacity, otherwise the server does not serve this file type.
\end{lemma}
\begin{proof}
In order to prove this, we show that the function:
\begin{align}
g(x) &= \sum p_2(y|x)V(y) - \sum p_1(y|x)V(y) \nonumber
\end{align}
is strictly decreasing.
On simplifying this expression, we get:
\begin{align}
g(x) = \mu_k(V(x-1) - V(x))
\end{align}
which is a strictly decreasing function in $x$ by Lemma \ref{lemma:inc_diff}. Thus the minimizer in (\ref{eq:Simple_DP}) changes from one to the other as this quantity crosses $\lambda_k$, while remaining fixed on either side thereof. This implies that the optimal policy is a threshold policy.
\end{proof}

 \textbf{Note: }We have made the assumption that the cost function $f$ is strictly convex. We can relax this assumption to mere convexity and get analogous statements of Lemma \ref{lemma_inc_file} and \ref{lemma:inc_diff}, except that increasing will be replaced by non-decreasing. The only difference it makes is that the choice of threshold, and therefore of our Whittle-like index, may become non-unique over a closed interval wherever the value function has a linear patch. This can be disambiguated by using the convention that we use the smallest candidate value as the index, i.e., the smallest value of the state $x$ for which it is equally desirable to be active or passive. It is easy to see that this is well defined and moreover, facilitates the ordinal comparisons in an unambiguous manner. Note that the scheduling policy depends only on such comparisons.  Thus this does not cause any inconsistency and remains a plausible heuristic, though it is not clear how the performance get affected vis-a-vis the case when such ambiguities do not arise. That it still is a reasonable heuristic is supported by our simulations on a linear cost function reported below. We may add that while strict convexity of the cost function $f$ ensures strict convexity of the value function $V$ as seen above, the latter may turn out to be strict convex even in cases where $f$ is not.

\section{Whittle-like Indexability}
\label{sec:whittle}

We next prove a Whittle-like indexability result in the spirit of \cite{Whittle}. We use the phrase `Whittle-like' because our problem formulation differs from that of \cite{Whittle}, though it builds upon it.

Let $\pi^{\ell}$ denote the stationary probability distribution when the threshold is $\ell$. That is,  if the number of jobs is $\leq {\ell}$, then the server does not transmit, and if number of jobs is $>{\ell}$, then the server transmits at full rate. We have the following lemma.
\begin{lemma}
\label{lemma:inc_stat}
$\sum_{i=0}^{{\ell}}\pi^{\ell}(i)$ is strictly increasing with ${\ell}$.
\end{lemma}
\begin{proof}
Let $\hat{\mu} = \sum_{j \in S_i; j \neq k}\mu_j$. The Markov chain formed with a threshold of ${\ell}$ is shown in  Figure \ref{fig:Markov_Model}. This is a time reversible Markov chain with stationary probabilities given by:
\begin{align}
\pi^{\ell}(i) &= \pi^{\ell}(0) \bigg( \frac{\Lambda}{\hat{\mu}} \bigg)^i \qquad \qquad \qquad \qquad \text{if } i \leq {\ell}, \nonumber \\
\pi^{\ell}(i) &= \pi^{\ell}(0) \bigg( \frac{\Lambda}{\hat{\mu}} \bigg)^{\ell} \bigg(\frac{\Lambda}{\mu_k + \hat{\mu}} \bigg)^{i-{\ell}} \qquad \quad \text{if } i > {\ell}, \nonumber
\end{align}
where $\pi^{\ell}(0)$ is the stationary probability of state $0$.

\begin{figure}[H]
\begin{center}
\includegraphics[angle=0,scale=0.4]{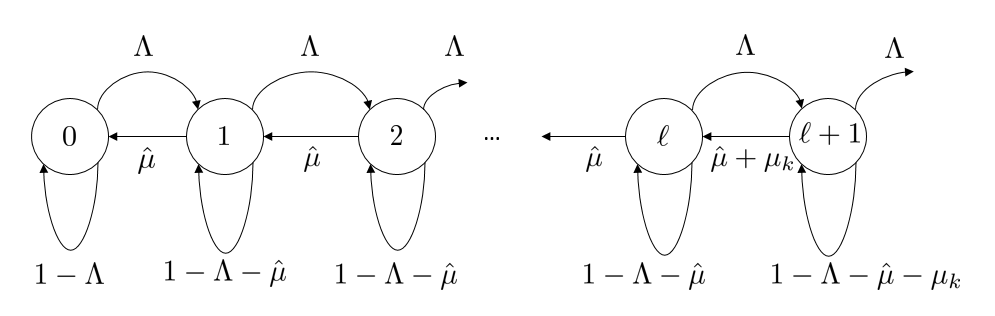}
\caption{Markov Chain}
\label{fig:Markov_Model}
\end{center}
\end{figure}

\noindent From this, we see that
\begin{align}
\sum_{i=0}^{{\ell}}\pi^{\ell}(i) = \frac{\frac{(\frac{\Lambda}{\hat{\mu}})^{\ell} - 1}{(\frac{\Lambda}{\hat{\mu}}) - 1}}{\frac{(\frac{\Lambda}{\hat{\mu}})^{\ell} - 1}{(\frac{\Lambda}{\hat{\mu}}) - 1} + (\frac{\Lambda}{\hat{\mu}})^{\ell} (\frac{\hat{\mu} + \mu_k}{\hat{\mu} + \mu_k - \Lambda} )} \nonumber
\end{align}
which is a strictly increasing function of ${\ell}$.
\end{proof}

\begin{theorem}
This problem is Whittle-like indexable in the sense that the set of passive states decreases monotonically from the whole state space to the empty set $\phi$ as $\lambda\uparrow\infty$.
\end{theorem}
\begin{proof}

The proof is along the lines of Theorem 1 in \cite{proc_sharing}. It has been reproduced for sake of completeness.

The optimal average cost of the problem is given by
$$\beta(\lambda) = \inf \: \{ \sum_k f(k) \pi(k) + \lambda\sum_{k \in B}\pi(k) \}$$
where $\pi$ is the stationary distribution and $B$ is the set of passive states. The infimum $\beta(\lambda)$ of this quantity affine in $\lambda$ is over all admissible policies, which by Lemma \ref{lemma:threshold_policy} is the same as the infimum over all threshold policies.  Hence $\beta(\cdot)$  is concave non-decreasing with slope $<1$. 
By the envelope theorem (Theorem 1, \cite{Milgrom}), the derivative of this function with respect to $\lambda$ is given by
$$\sum_{k=0}^{x(\lambda)}\pi^{x(\lambda)}(k)$$
where $x(\lambda)$ is the optimal threshold under $\lambda$. 
Since $\beta (\lambda)$ is a concave function, its derivative has to be a non-increasing function of $\lambda$, i.e.,
$$\sum_{k=0}^{x(\lambda)}\pi^{x(\lambda)}(k) \: \: \: \text{is non-increasing with} \: \lambda.$$
But, from Lemma \ref{lemma:inc_stat}, we know that $\sum_{j=0}^{{\ell}}\pi^{{\ell}}(j)$ is a strictly increasing function of ${\ell}$, where ${\ell}$ is the threshold.
Then $x(\lambda)$ must be a strictly decreasing function of $\lambda$. The set of passive states for $\lambda$ is given by
$ [0, x(\lambda)]$. It follows 
 that the set of passive states monotonically decreases to $\phi$ as $\lambda \uparrow \infty$. This implies Whittle-like indexability.

\end{proof}

\subsection{Proposed Policy}

We propose the following heuristic policy inspired by \cite{Whittle}:\\

\textit{Our decision epochs are the time instances when there is some change in the system, i.e., either an arrival or a departure occurs.
For each server $j \in S$, the index computed for each file type connected to this server is known. The file type which has the smallest index is chosen and the server serves it at full rate. Each time there is either an arrival into a file type or there is a job completion, the new indices are sent to the server which then decides which queue to serve.}

\subsection{Computation of the Whittle-like index}

The Whittle-like index $\lambda(x)$ (when the number of jobs is $x$) is computed by the following linear system of equations and an iterative scheme for $\lambda_n$ which uses the solution of the linear system as a subroutine at each step. We have used $V_{\lambda}(\cdot), \beta(\lambda)$ in place of $V(\cdot), \beta$ to make the $\lambda$-dependence of $V, \beta$ explicit as required by this part of analysis.
\begin{align}
V_{\lambda_n}(y) &= f(y) + \lambda_n(x) + \mathbb{E}_{p_1}[ V_{\lambda_n}(y)] - \beta(\lambda_n) \: \:  \text{if } y \leq x, \label{eq:Value_Change_0} \\
V_{\lambda_n}(y) &= f(y) +  \mathbb{E}_{p_2}[V_{\lambda_n}(y)] - \beta(\lambda_n) \: \qquad   \: \: \quad \text{if } y > x,
\label{eq:Value_Change_1} \\
V_{\lambda_n}(0) &= 0 \label{eq:Value_Change2} \\
\lambda_{n+1}(x) &= \lambda_n + \eta(\mathbb{E}_{p_2}[V_{\lambda_n}(x)] - \mathbb{E}_{p_1}[ V_{\lambda_n}(x)] - \lambda_n(x)) \label{eq:Lambda_Change_1}
\end{align}
Here $\eta$ is a small step size (taken to be $0.01$). $\mathbb{E}_{p_i} [ \  \cdot \ ]$ denotes expectation with respect to the probability distribution $p_i, i = 1,2$ (as defined in equations (\ref{eq:def_passive_prob}) and (\ref{eq:def_active_prob})).

We analyze this scheme under the simplifying assumption that $f$ is strictly convex. The proof of convexity of $V$ shows that $V$ will also be strictly convex, hence $x \mapsto V(x + z) - V(x)$ for $z > 0$ strictly increasing. In particular, the argument of Lemma \ref{lemma:threshold_policy} then shows that the Whittle-like index is uniquely defined for each $x$.

\begin{theorem} For each fixed $x$, $\lambda_n(x)$ converges to an $O(\eta)$ neighborhood of the Whittle-like index as $n\uparrow\infty$.
\end{theorem}

\begin{proof} Since $\eta$ is small, we can view (\ref{eq:Lambda_Change_1}) as an Euler scheme for approximate solution by discretization \cite{Butcher} of the ODE
\begin{align}
\dot{\lambda}(t) = F(\lambda(t)) - \lambda(t) \nonumber
\end{align}
where
$$F(\lambda) := \mathbb{E}_{p_2}[V_{\lambda}(x)] - \mathbb{E}_{p_1}[ V_{\lambda}(x)]. $$
 Equations (\ref{eq:Value_Change_0}) - (\ref{eq:Value_Change2}) constitute a linear system of equations, hence $V_{\lambda}(x), \beta(\lambda)$ are linear in $\lambda$. Thus the above ODE is well-posed.
Furthermore, this is a scalar ODE  with  equilibrium given by that value of $\lambda$ for which
\begin{align}
\lambda = \mathbb{E}_{p_2}[V_{\lambda}(x)] - \mathbb{E}_{p_1}[ V_{\lambda}(x)], \nonumber
\end{align}
i.e., the Whittle-like index at state $x$, unique as observed above.  Above this value, the ODE has a negative drift and below it, a positive drift. Thus it is a stable ODE (i.e., the trajectories do not blow up).  As  a stable linear ODE, it converges to its equilibrium.
 Interpolate the iterates as $\bar{\lambda}(t) = \lambda(n)$ for $t = n\eta$ with linear interpolation on $[n\eta, (n+1)\eta]$ $\forall n$. Define $[t] := \sup\{n\eta : n\eta \leq t < (n+1)\eta\}$. Then we have
 $$\dot{\bar{\lambda}}(t) = F(\bar{\lambda}(t)) - \bar{\lambda}(t) +  \upsilon(t) \quad  \mbox{a.e.},$$
 where $\upsilon(t) := F(\bar{\lambda}([t])) - F(\bar{\lambda}(t))$.
It is easy to check that given the boundedness of trajectories and linearity of $F$, $|\upsilon|$ is $O(\eta)$. The convergence of iteration (\ref{eq:Lambda_Change_1}) to a neighborhood of this equilibrium then follows from Theorem 1 of \cite{Hirsch} by standard arguments. In fact, given that this is a linear system with input, the classical variation of constants formula can be used for the purpose as well. (See \cite{Butcher} for a detailed error analysis of Euler method in a much more general set-up.)
\end{proof}

Note that we have not imposed any restriction on the sign of $\lambda(x)$ though it is known a priori, because the stable dynamics above with a unique equilibrium automatically picks up the right $\lambda(x)$. The linear system (\ref{eq:Value_Change_0}), (\ref{eq:Value_Change_1}) and
(\ref{eq:Value_Change2}) is solved as a subroutine by a suitable linear system solver.

\section{Simulations}
\label{sec:simulations}

In this section, we report simulations to compare the performance of the proposed Whittle-like index policy with other natural heuristic policies given as follows:

 \begin{itemize}
\item \textbf{Balanced Fair Allocation: } This is a centralized scheme for allocating server capacities. See \cite{BF} for more details.
\item \textbf{Uniform Allocation: } At each instant in time, each server splits its rate equally among all the files that it contains.
\item \textbf{Weighted Allocation: }The server rates are split according to prescribed weights proportional to the arrival rates into the different file types.
\item \textbf{Random Allocation: } The decision epochs are the same. At each instant, for each server, a file type is chosen randomly and the server serves this at full capacity.
\item \textbf{Max-Weight Allocation: } Each server serves at full capacity that file type which has the most number of jobs at any given instant.
\end{itemize}

We first compare the Whittle-like policy with the Optimal scheme and the balanced fairness allocation scheme for the network shown in Figure \ref{fig:Network_Sim_Opt}. The results of the simulations are shown in Figure \ref{fig:Sim_Opt}. The values of parameters are as follows:
$\Lambda^1 = 0.2 ;  f^1(x) = 13x ; \Lambda^2 = 0.1 ; f^2(x) = 10x ; \mu_1 = 0.2 ; \mu_2 = 0.2$.
The optimal value for the network in figure \ref{fig:Network_Sim_Opt} is computed using the following set of iterative equations:
\begin{align}
V_{n+1}(x_1, x_2) &= f^1(x_1) + f^2(x_2) - V_n(0,0) \nonumber \\
&+ \min_{i \in \{ 1,2,3,4 \} } \big( \mathbb{E}^i[V_n(\cdot) | x_1,x_2] \big) \nonumber \\
u_{n+1} &= \underset{i}{\operatorname{argmin}} \big( \mathbb{E}^i[V_n(\cdot) | x_1,x_2] \big) \nonumber
\end{align}
Here, the control $u$ denotes the following: $u=1$ denotes server $1,2$ serve file $1$; $u=2$ denotes server $1$ serves file $1$ and server $2$ serves file $2$; $u=3$ denotes server $1$ serves file $2$ and server $2$ serves file $1$; $u=4$ denotes server $1,2$ serve file $2$. $\mathbb{E}^i[ \ \cdot \ ]$ denotes expectation with respect to the probability distribution under control $i$.

\begin{figure}[h!]
\begin{center}
\includegraphics[angle=0,scale=0.6]{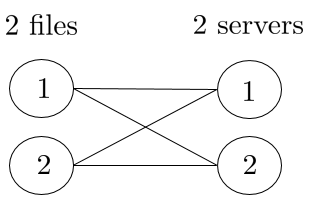}
\caption{The Network that we use for simulations}
\label{fig:Network_Sim_Opt}
\end{center}
\end{figure}

\begin{figure}[h!]
\begin{center}
\includegraphics[angle=0,scale=0.5]{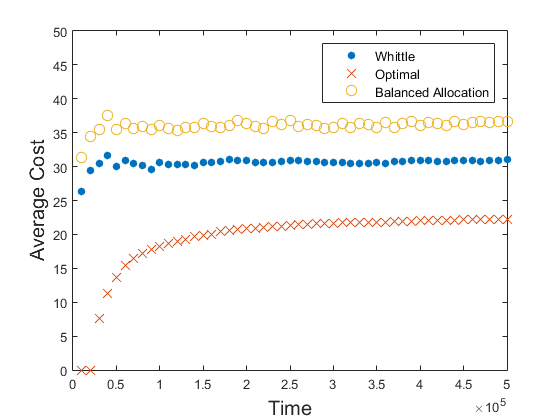}
\vspace{-3mm}
\caption{Comparison of Whittle-like policy, True Optimal, and Balanced Fairness scheme}
\label{fig:Sim_Opt}
\end{center}
\end{figure}

The second network that we  consider is shown in Figure \ref{fig:Network_Sim}. The  parameters in this simulation are as follows:
$\Lambda^1 = 0.1 ; f^1(x) = 10x ; \Lambda^2 = 0.2 ; f^2(x) = 20x ; \Lambda^3 = 0.1 ; f^3(x) = 10x ; \mu_1 = 0.2 ; \mu_2 = 0.3$.

\begin{figure}[h!]
\begin{center}
\includegraphics[angle=0,scale=0.6]{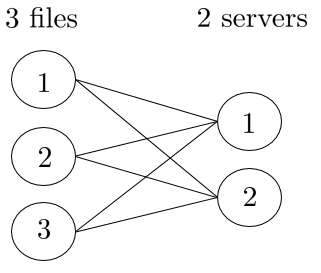}
\vspace{-3mm}
\caption{The Network that we use for simulations}
\label{fig:Network_Sim}
\end{center}
\end{figure}

Figure \ref{fig:Whittle_Files} shows the Whittle-like indices assigned by file types $1$ and $2$ to server $1$ and Figure \ref{fig:Whittle_Servers} shows the Whittle-like indices assigned by file type 1 to the two servers.

\begin{figure}[H]
\begin{center}
\includegraphics[angle=0,scale=0.5]{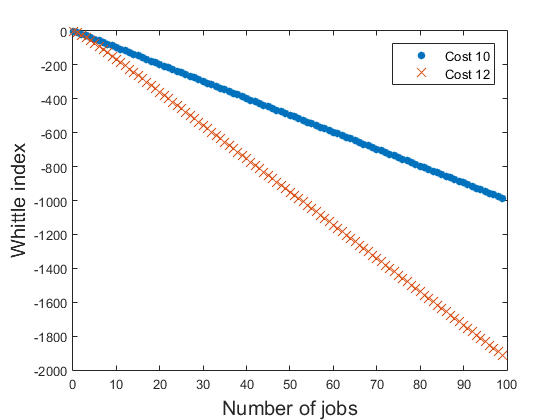}
\vspace{-3mm}
\caption{Whittle-like index assigned by different files to the same server}
\label{fig:Whittle_Files}
\end{center}
\end{figure}

\begin{figure}[h!]
\begin{center}
\includegraphics[angle=0,scale=0.5]{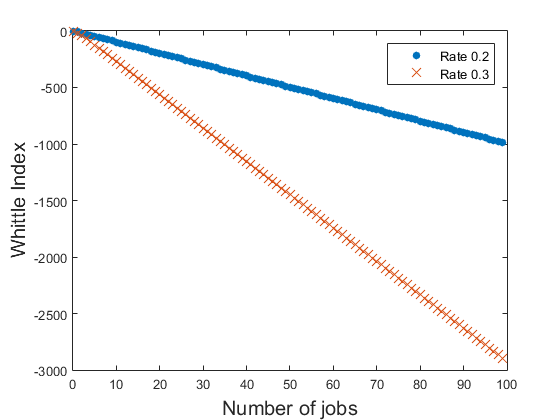}
\vspace{-3mm}
\caption{Whittle-like index assigned by the same file to different servers}
\label{fig:Whittle_Servers}
\end{center}
\end{figure}

Figures \ref{fig:Sim_Bad} and \ref{fig:Sim_Good} compare performance of the various methods that were described earlier in this section\footnote{We have separated these figures for better comparison. This is because the performance of the uniform and random allocation is much worse than the other policies.}. We can see that the Whittle-like index based policy performs better than the other methods of server allocation.

\begin{figure}[h!]
\begin{center}
\includegraphics[angle=0,scale=0.5]{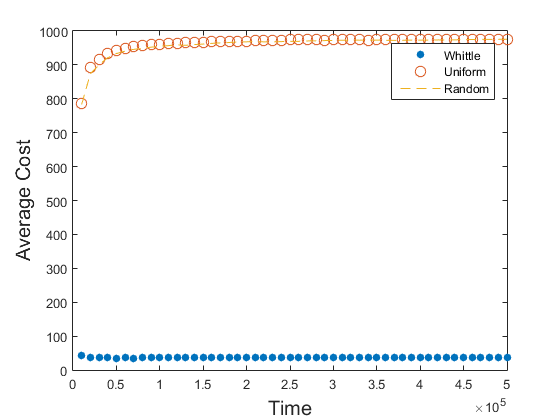}
\vspace{-3mm}
\caption{Comparison of Whittle-like policy with Uniform and Random policies }
\label{fig:Sim_Bad}
\end{center}
\end{figure}

\begin{figure}[h!]
\begin{center}
\includegraphics[angle=0,scale=0.5]{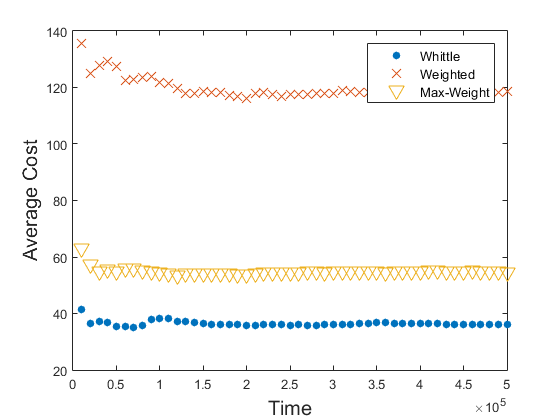}
\vspace{-3mm}
\caption{Comparison of Whittle-like policy with Weighted and Max Weight policies }
\label{fig:Sim_Good}
\end{center}
\end{figure}
%
%
%
Figure \ref{fig:Sim_Good_1} shows simulation results for the model with 10 file types and 10 servers such that file type $i$ is stored in servers $i, i+1 (\text{mod} 10)$. $\Lambda^i = 0.2, f^i(x) = 15x, \mu_i = 0.2$ for $i = 1, 4, 7, 10$.  $\Lambda^i = 0.3, f^i(x) = 20x, \mu_i = 0.3$ for $i = 2, 5, 8$. $\Lambda^i = 0.1, f^i(x) = 10x, \mu_i = 0.2$ for $i = 3, 6, 9$. Again, the Whittle-like policy shows a clear advantage.

\begin{figure}[h!]
\begin{center}
\includegraphics[angle=0,scale=0.5]{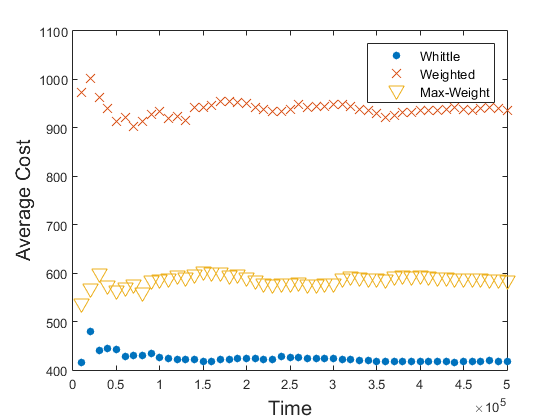}
\vspace{-3mm}
\caption{Comparison of Whittle-like policy with Weighted and Max Weight policies ($10$ file types and $10$ servers)}
\label{fig:Sim_Good_1}
\end{center}
\end{figure}
\section{Conclusions and Future Work}
\label{sec:conc}
\noindent We have proved Whittle-like indexability of the server allocation problem in resource pooling networks. The allocation of servers using the Whittle-like scheme can be implemented in a distributed manner. The next step would be to extend this work to more general file types and possibly more complicated network topologies.

\end{document}